\documentclass[conference]{IEEEtran}
\usepackage{cite,graphicx,amssymb,amsmath,color,textcomp}
\usepackage{amsmath}   
\usepackage{amssymb}   
\usepackage{extarrows,multirow,multicol}
\usepackage{bm}
\usepackage{float}
\usepackage{subfig}

\newtheorem{remark}{\underline{Remark}}
\newtheorem{lemma}{Lemma}

\newtheorem{theorem}{Theorem}

\newlength{\figwidth}
\IEEEoverridecommandlockouts

\begin{document}

\title{Standard Condition Number–Based Detection for MIMO ISAC Systems under Noise Uncertainty
\vspace{-0.2cm}
}
\author{
 \IEEEauthorblockN{Alex Obando, Tharindu Udupitiya,  
 Saman Atapattu and Kandeepan Sithamparanathan\\  
 }
 \IEEEauthorblockA{
Department of Electrical and Electronic Engineering, RMIT University, Melbourne, Victoria, Australia\\
\IEEEauthorblockA{\{alex.daniel.obando,\,tharindu.udupitiya,\,saman.atapattu,\,kandeepan.sithamparanathan\}@rmit.edu.au
}
}
\vspace{-0.5cm}

}

\maketitle


\begin{abstract}
This paper presents a unified analytical and optimization framework for \emph{Standard Condition Number} (SCN)–based detection in MIMO Integrated Sensing and Communication (ISAC) systems operating under noise uncertainty.  
Conventional detectors such as the Likelihood Ratio Test (LRT) and Energy Detector (ED) suffer from false-alarm inflation when interference or jamming alters the noise covariance.  
To overcome this limitation, the SCN detector, defined as the ratio of the largest to smallest eigenvalues of the sample covariance matrix is analytically characterized for the first time in an ISAC setting.  
Closed-form expressions for the false-alarm and detection probabilities are derived using random matrix theory for a two-antenna sensing receiver and generalized to arbitrary MIMO dimensions.  
The analysis proves that the SCN maintains a \emph{constant false alarm rate} (CFAR) property and remains resilient to covariance mismatch, providing theoretical justification for its robustness in dynamic environments.  
Leveraging these results, a tractable ISAC power-allocation problem is formulated to minimize total detection error subject to communication rate and power constraints, yielding an interpretable sequential solution.  
Numerical evaluations verify the theory and demonstrate that the proposed SCN detector consistently outperforms LRT and eigenvalue-based benchmarks, particularly under strong interference and jamming typical of modern multiuser networks.
\end{abstract}


\vspace{-6mm}

\section{Introduction} \label{S1}
\vspace{-1mm}

Integrated sensing and communication (ISAC) has emerged as a core paradigm for next-generation wireless networks, enabling joint radar sensing and data transmission over a shared spectrum~\cite{zhang2021isac,cheng2023coordinated}.  
By leveraging the dual functionality of radio signals, ISAC enables positioning, autonomous driving, and environmental monitoring.~\cite{liang2022joint}.  
Modern MIMO ISAC systems use dual precoders and unified waveforms to balance sensing accuracy and communication throughput via dynamic power allocation.~\cite{9724174,10938928,8239836}.  
Existing approaches, including sensing-embedded communication, communication-embedded radar, and joint co-design, present trade-offs between complexity and flexibility.~\cite{liang2022joint,Perera2025opportunistic}.  
However, robust ISAC operation under dynamic interference and noise uncertainty remains challenging, motivating the need for statistically resilient detection and power-allocation frameworks.

Thus, signal detection is fundamental to reliable target identification and situational awareness in ISAC networks.  
Conventional detectors such as the likelihood ratio test (LRT), generalized LRT (GLRT), and energy detector (ED) rely on accurate knowledge of signal and noise statistics~\cite{an2023TWC,lou2024power,zhao2024dual}, which is rarely available in dynamic or interference-limited environments.  
Eigenvalue-based detectors including the maximum eigenvalue, trace-to-maximum ratio, and condition-number tests address this limitation by exploiting the covariance structure of the received data, enabling blind detection without explicit noise-power estimation~\cite{7903598,Chamain,debbah2011,nafkaSCN2020}.  
However, their integration into ISAC frameworks remains largely unexplored~\cite{obando2025eigenvalue}, particularly under time-varying noise and jamming conditions.

A major limitation of existing detectors is the assumption that estimated parameters such as noise covariance or channel statistics remain constant during the sensing process~\cite{an2023TWC,lou2024power,zhao2024dual,7903598,Chamain,debbah2011,nafkaSCN2020}.  
In practice, these parameters fluctuate due to jamming, interference, and hardware impairments, rendering conventional detectors highly sensitive to uncertainty \cite{udupitiya2024}.  
Although recent works have introduced bounded CSI error models and robust waveform optimization under statistical or semi-infinite constraints, noise uncertainty particularly covariance mismatch arising in non-stationary environments remains largely unaddressed~\cite{besson2016generalized,werner2007,johnstone2020,10619207}.  
The \emph{Standard Condition Number} (SCN) detector, defined as the ratio of the largest to smallest eigenvalues of the covariance matrix, inherently exhibits scale-invariant behavior~\cite{nafkaSCN2020,udupitiya2024}, making it naturally robust to unknown noise power.  
Hence, extending SCN-based detection to ISAC systems under noise uncertainty constitutes a critical and unexplored research.

This paper fills this critical gap by presenting the \emph{first rigorous analysis of Standard Condition Number (SCN)–based detection in MIMO ISAC systems} operating under noise uncertainty.  
The proposed framework integrates theoretical derivations with practical ISAC design insights and is broadly applicable to general detection problems beyond ISAC.  
The key contributions are summarized as follows:  
(i) closed-form analytical expressions are derived for the SCN detection probability under unknown and time-varying noise covariance;  
(ii) a unified optimization framework is formulated to allocate transmit power between communication and sensing functions for robust ISAC operation under uncertainty; and  
(iii) the impact of noise mismatch on the sensing–communication tradeoff is analyzed, revealing that SCN preserves the constant false alarm rate (CFAR) property and achieves superior reliability compared with conventional LRT and maximum-eigenvalue detectors.  
To the best of our knowledge, this is the first comprehensive theoretical and practical treatment of SCN detection for MIMO ISAC, establishing a foundation for robust spectrum sensing and joint communication–sensing design in next-generation wireless networks.

\begin{figure*}[t!]
    \centering
    \subfloat[Noise-only observation for covariance estimation.\label{fig_sys_mod_estim}]{
        \includegraphics[width=0.31\textwidth]{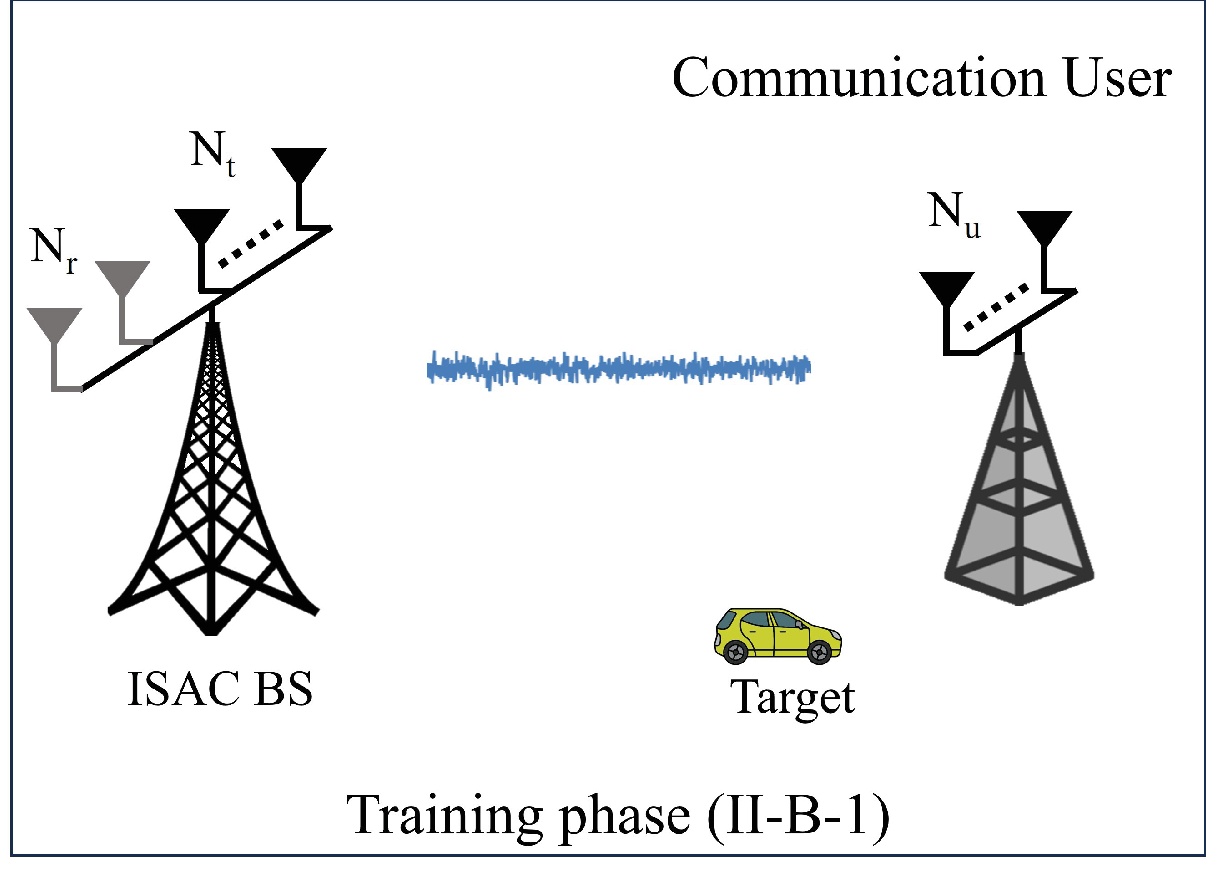}
    }\hfill
    \subfloat[Joint ISAC signal reception.\label{fig_sys_sens}]{
        \includegraphics[width=0.31\textwidth]{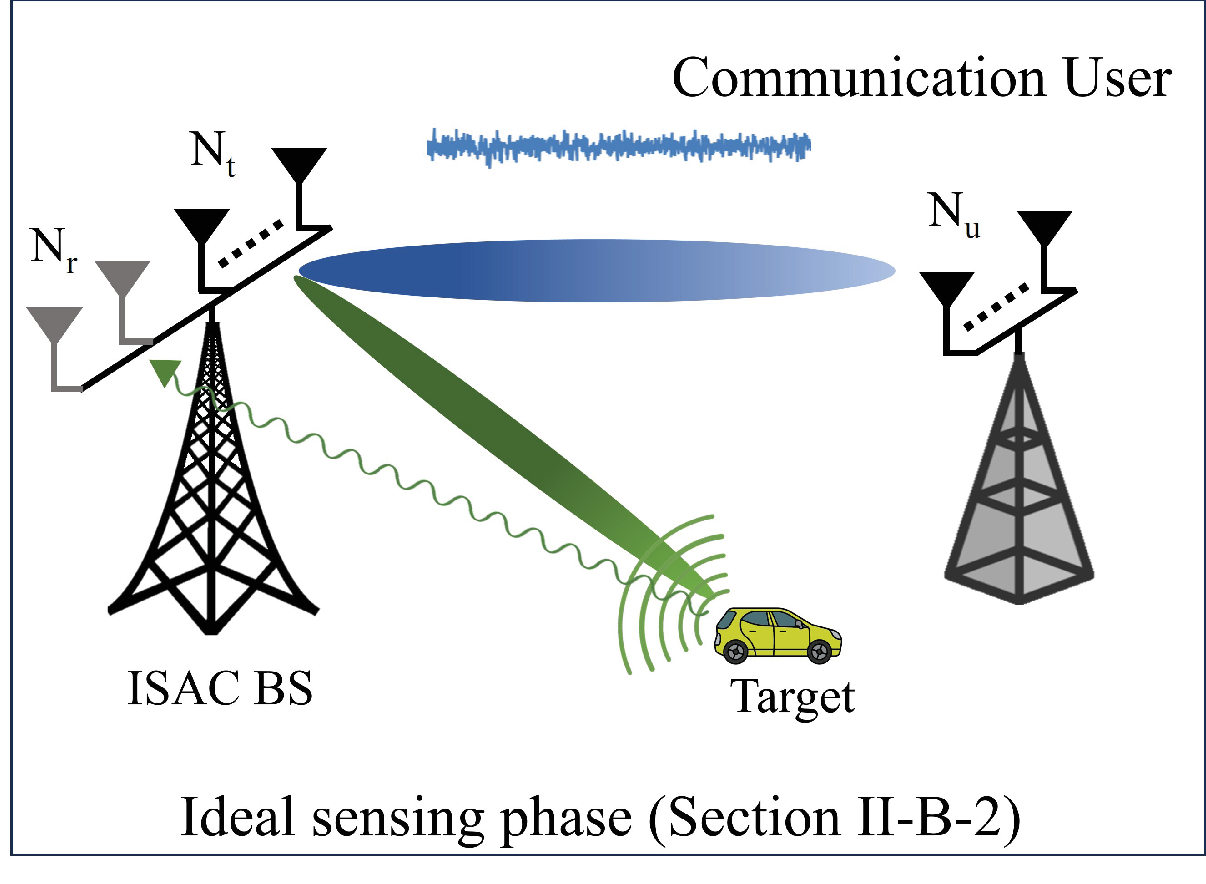}
    }\hfill
    \subfloat[Sensing under interference/jamming.\label{fig_sys_jam}]{
        \includegraphics[width=0.31\textwidth]{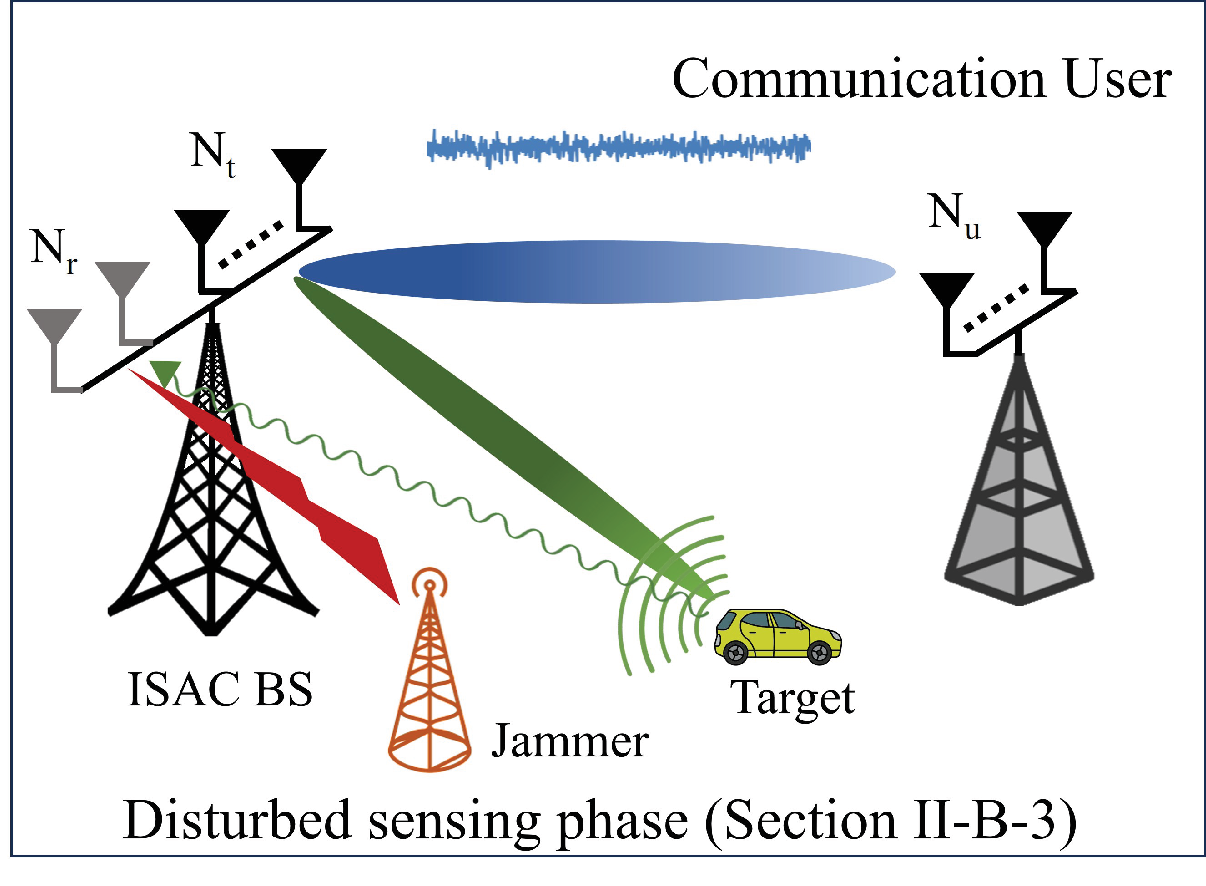}
    }
    \caption{The proposed ISAC framework showing (a) training, (b) ideal sensing, and (c) disturbed sensing under noise uncertainty.}
    \vspace{-4mm}
\end{figure*}



\section{Network and Signal Model}\label{net_mod}
We consider a monostatic ISAC system where a MIMO base station (BS) equipped with $N_t$ transmit  antennas serves a single multi-antenna communication user (CU) with \(N_u\) receive antennas while simultaneously sensing a target using an \(N_r\)-antenna receive array.  
A dual-waveform transmission strategy is adopted, in which dedicated waveforms are assigned for communication and sensing. 

\subsection{Communication Model}
At time slot $\ell$, the BS transmits a superposition of communication and sensing signals. Let $\mathbf{s}_\mathrm{c}[\ell]\in\mathbb{C}^{N_u\times1}$ denotes the communication vector and $s_\mathrm{s}[\ell]\!\in\!\mathbb{C}$ denotes the sensing symbol, satisfying $\mathbb{E}[\mathbf{s}_\mathrm{c}[\ell]\mathbf{s}_\mathrm{c}^H[\ell]]=\mathbf{I}_{N_u}$, $\mathbb{E}[|s_\mathrm{s}[\ell]|^2]=1$, and $\mathbb{E}[s_\mathrm{c}[\ell]s_\mathrm{s}^*[\ell]]=0$. The corresponding precoders are $\mathbf{W}_\mathrm{c}\!\in\!\mathbb{C}^{N_t\times N_u}$ and $\mathbf{w}_\mathrm{s}\!\in\!\mathbb{C}^{N_t\times1}$. The transmitted signal at time $\ell$ is given by \(\mathbf{x}[\ell]
=\mathbf{W}_\mathrm{c}\mathbf{s}_\mathrm{c}[\ell]
+\mathbf{w}_\mathrm{s}s_\mathrm{s}[\ell].\)
Over $L$ time slots, the transmit matrix becomes
\begin{equation}\label{eq:tx_signal}
\mathbf{X}
=\mathbf{W}_\mathrm{c}\mathbf{S}_\mathrm{c}
+\mathbf{w}_\mathrm{s}\mathbf{s}_\mathrm{s}
\in\mathbb{C}^{N_t\times L},
\end{equation}
where $\mathbf{S}_\mathrm{c}=[\mathbf{s}_\mathrm{c}[1],\ldots,\mathbf{s}_\mathrm{c}[L]]\in \mathbb{C}^{N_u \times L}$ and $\mathbf{s}_\mathrm{s}=[s_\mathrm{s}[1],\ldots,s_\mathrm{s}[L]] \in \mathbb{C}^{1 \times L}$. The sensing waveform $\mathbf{s}_\mathrm{s}$ is assumed to be known at both the BS and the CU receiver, enabling interference cancellation. 

\subsubsection*{Received Communication Signal}
The received baseband signal at the CU with $N_u$ antennas is \(\mathbf{y}_\mathrm{c}[\ell]=\mathbf{H}_\mathrm{c}\mathbf{W}_\mathrm{c}\mathbf{s}_\mathrm{c}[\ell]
+\mathbf{H}_\mathrm{c}\mathbf{w}_\mathrm{s}s_\mathrm{s}[\ell]
+\mathbf{n}_\mathrm{c}[\ell],\)
where $\mathbf{H}_\mathrm{c}\!\in\!\mathbb{C}^{N_u\times N_t}$ represents the BS–CU MIMO channel with i.i.d.\ entries following $\mathcal{CN}(0,\sigma_h^2)$, and $\mathbf{n}_\mathrm{c}[\ell]\!\sim\!\mathcal{CN}(\mathbf{0},\sigma_c^2\mathbf{I}_{N_u})$ denotes the additive white Gaussian noise (AWGN) vector. Since the sensing waveform is known to the CU, the interference term $\mathbf{H}_\mathrm{c}\mathbf{w}_\mathrm{s}s_\mathrm{s}[\ell]$ can be mitigated at the receiver.

The instantaneous signal-to-noise ratio (SNR) for the CU becomes \(\mathrm{SNR}_{\mathrm{c}}
=\|\mathbf{H}_\mathrm{c}\mathbf{W}_\mathrm{c}\|^2_F/\sigma_c^2.\) Since the BS possesses only statistical channel state information (CSIT),
the achievable \emph{ergodic rate} is
\begin{equation}\label{eq:ergodic_rate}
\bar R
= \mathbb{E}_{\mathbf{H}_{\mathrm{c}}}\!\left[
  \log_2\!\left(1 + \tfrac{\|\mathbf{H}_\mathrm{c}\mathbf{W}_\mathrm{c}\|^2}{\sigma_c^2}\right)
\right]= \frac{1}{\ln 2} \, e^{1/\rho} 
    \sum_{m=1}^{N_u} E_m\left(\frac{1}{\rho}\right)
\end{equation}

where \(\rho = \sigma_h^2\|\mathbf{W}_\mathrm{c}\|^2/\sigma_c^2\) and \(E_n(\cdot)\) is the exponential integral of order $n$.
This ergodic formulation aligns with the statistical-CSIT assumption and characterizes the long-term communication throughput.

\subsection{Sensing Model}

In the considered ISAC system, the BS simultaneously transmits communication and sensing waveforms while receiving target echoes through its sensing array.  
The overall sensing process comprises three operational phases: training, ideal sensing, and disturbed sensing.

\subsubsection{Training Phase}
During a silent interval, the BS transmits pilots to collect $L$ noise-only snapshots used for covariance estimation using received signal 
\(
\mathbf{y}_s[\ell] = \mathbf{n}_s[\ell]\), with 
\(\mathbf{n}_s[\ell] \sim \mathcal{CN}_{N_r}(\mathbf{0}, \sigma_s^2\mathbf{I}_{N_r}),\) for \(\ell=1,\ldots,L.\)
These samples yield an estimate of the nominal noise covariance $\sigma_s^2\mathbf{I}_{N_r}$ for subsequent sensing.

\subsubsection{Ideal Sensing Phase}
During active sensing, the BS transmits $\mathbf{x}[\ell]$ and receives 
\begin{equation}
\mathbf{y}_{\mathrm{d}}[\ell]=
\begin{cases}
\mathbf{n}_s[\ell], & \mathcal{H}_0,\\[2pt]
\mathbf{G}\mathbf{x}[\ell]+\mathbf{n}_s[\ell], & \mathcal{H}_1,
\end{cases}
\label{eq:idel_sensing_phase}
\end{equation}
where $\mathbf{G}\!\in\!\mathbb{C}^{N_r\times N_t}$ denotes the composite BS–target–BS channel and $\mathbf{n}_s[\ell]$ retains the same covariance as in training.  
This represents the \emph{ideal} stationary case where training and sensing share identical noise statistics.

\subsubsection{Disturbed Sensing Phase}
In practice, external interference or intentional jamming alters the sensing covariance.  
Let $\mathbf{j}[\ell]\!\sim\!\mathcal{CN}_{N_r}(\mathbf{0},\sigma_j^2\mathbf{I}_{N_r})$ denote the disturbance term.  
The received signal is then
\begin{equation}
\mathbf{y}_{\mathrm{d}}[\ell]=
\begin{cases}
\mathbf{n}_s[\ell]+\mathbf{j}[\ell], & \mathcal{H}_0,\\[2pt]
\mathbf{G}\mathbf{x}[\ell]+\mathbf{n}_s[\ell]+\mathbf{j}[\ell], & \mathcal{H}_1,
\end{cases}
\label{eq:disturbed_sensing_phase}
\end{equation}
where $\sigma_j^2=\mu'\sigma_s^2$ and $\mu=1+\mu'$ is the \textit{covariance mismatch factor}.  
The effective sensing-phase covariance becomes $\mu\sigma_s^2\mathbf{I}_{N_r}$. 
Stacking $L$ snapshots gives
\begin{equation}
\mathbf{Y}_\mathrm{d}=\mathbf{G}\mathbf{X}+\mathbf{N}+\mathbf{J}
=\mathbf{G}\mathbf{W}\mathbf{S}+\mathbf{N}+\mathbf{J},
\label{eq:sensing_matrix_model}
\end{equation}
where $\mathbf{N}$ and $\mathbf{J}$ are the stacked noise and interference matrices, and $\mathbf{W}=[\mathbf{W}_\mathrm{c},\mathbf{w}_\mathrm{s}]$, $\mathbf{S}=[\mathbf{S}_\mathrm{c}^T,\mathbf{s}_\mathrm{s}^T]^T$ describe the joint transmit structure. 
The target response is modeled as $\mathbf{G}=\beta\,\mathbf{a}(\theta)\mathbf{b}^H(\theta)$, where $\beta\!\in\!\mathbb{C}$ represents the complex reflection coefficient, and $\mathbf{a}(\theta)$ and $\mathbf{b}(\theta)$ are the receive and transmit steering vectors.  
For ULAs with half-wavelength spacing,
\(
\mathbf{a}(\theta)=[1,e^{-j\pi\sin\theta},\ldots,e^{-j\pi(N_r-1)\sin\theta}]^T\) and \(
\mathbf{b}(\theta)=[1,e^{-j\pi\sin\theta},\ldots,e^{-j\pi(N_t-1)\sin\theta}]^T.
\)
\vspace{0.2cm}
\begin{remark}[Novelty of Unified Sensing Framework]
Unlike conventional ISAC models that assume known or stationary noise statistics~\cite{7903598,an2023TWC,lou2024power,zhao2024dual,Chamain,debbah2011,nafkaSCN2020}, this work establishes a three-phase formulation as training, ideal, and disturbed, linking covariance estimation and robustness analysis within a single unified framework.  
To the best of our knowledge, this is the first ISAC formulation to explicitly model both matched ($\mu=1$) and mismatched ($\mu\neq1$) sensing regimes.
\end{remark}
\vspace{0.2cm}
\subsubsection{Motivation for SCN-Based Detection} 
Conventional detectors such as the LRT and $\lambda_{\max}$ test rely on precise noise knowledge, so even mild mismatches bias thresholds, increase false alarms, and degrade sensing reliability. These detectors also lose invariance to nuisance parameters, violating the CFAR property essential for reliable ISAC operation.  
To address this, the \textit{Standard Condition Number (SCN)} detector employs a ratio-based statistic between maximum and minimum eigenvalues, maintaining scale invariance and mitigating noise uncertainty to ensure stable sensing–communication coexistence.
\vspace{0.2cm}
\section{Standard Condition Number (SCN) Detector}

Eigenvalue-based detection relies on the spectral characteristics of the received covariance matrix.  
Under the disturbed sensing model in~\eqref{eq:sensing_matrix_model}, the population-level covariance is
\begin{align}
\boldsymbol{\Sigma}_\mathrm{d} &=
\begin{cases}
\mu\sigma_s^2\mathbf{I}_{N_r}, & \mathcal{H}_0,\\[2pt]
\mathbf{G}\mathbf{W}\mathbf{W}^H\mathbf{G}^H + \mu\sigma_s^2\mathbf{I}_{N_r}, & \mathcal{H}_1,
\end{cases}
\label{eq:covariance_model}
\end{align}
where $\mu$ is the \textit{covariance–mismatch factor}, and $\mu=1$ corresponds to the ideal sensing condition as in \eqref{eq:idel_sensing_phase}.  
The sample covariance estimated from the received data is
\begin{equation}
\widehat{\boldsymbol{\Sigma}}
= \frac{1}{L}\mathbf{Y}_{\mathrm{d}}\mathbf{Y}_{\mathrm{d}}^H
= \frac{1}{L}\sum_{\ell=1}^{L}\mathbf{y}_{\mathrm{d}}[\ell]\mathbf{y}_{\mathrm{d}}^H[\ell],
\label{eq:sample_covariance}
\end{equation}
which asymptotically approaches~\eqref{eq:covariance_model} as $L \to \infty$. 

The SCN detector exploits the eigenvalue spread of $\widehat{\boldsymbol{\Sigma}}$ to measure deviations from spatially white noise.  
The test statistic is defined as
\begin{equation}
\kappa(\widehat{\boldsymbol{\Sigma}})
\triangleq \frac{\lambda_{\max}(\widehat{\boldsymbol{\Sigma}})}{\lambda_{\min}(\widehat{\boldsymbol{\Sigma}})}.
\label{eq:scn_def}
\end{equation}
Under $\mathcal{H}_0$, all eigenvalues are approximately equal, yielding $\kappa\!\approx\!1$, whereas under $\mathcal{H}_1$, the signal subspace induced by $\mathbf{G}\mathbf{W}$ inflates the dominant eigenvalue, producing $\kappa\!>\!1$.

Statistically, the sample covariance follows
\begin{equation}
L\,\widehat{\boldsymbol{\Sigma}} \sim
\begin{cases}
\mathcal{CW}_{N_r}(L,\,\mu\sigma_s^2\mathbf{I}_{N_r}), & \mathcal{H}_0,\\[2pt]
\mathcal{CW}_{N_r}\!\left(L,\,\mu\sigma_s^2\mathbf{I}_{N_r},\,\mathbf{G}\mathbf{W}\mathbf{W}^H\mathbf{G}^H\right), & \mathcal{H}_1,
\end{cases}
\label{eq:scn_cov_model}
\end{equation}
where $\mathcal{CW}_{N_r}(L,\mathbf{\Sigma},\mathbf{\Theta})$ denotes a complex non-central Wishart distribution with $L$ degrees of freedom, scale matrix $\mathbf{\Sigma}$, and non-centrality parameter $\mathbf{\Theta}$.  


\begin{lemma}[CFAR Invariance of SCN]\label{lemma:1}
Under $\mathcal{H}_0$, scaling the covariance by any factor $\mu>0$ uniformly scales all eigenvalues, 
\(\lambda_i^{(\mu)} = \mu\,\lambda_i^{(1)},\, i=1,\ldots,N_r.\)  
Hence,
\begin{equation}
\kappa^{(\mu)}(\widehat{\boldsymbol{\Sigma}})
={\lambda_{\max}^{(\mu)}}/{\lambda_{\min}^{(\mu)}}
={\lambda_{\max}^{(1)}}/{\lambda_{\min}^{(1)}}
=\kappa^{(1)}(\widehat{\boldsymbol{\Sigma}}),
\end{equation}
showing that the SCN statistic is invariant to any global covariance scaling.  
Consequently, the SCN detector inherently achieves the \emph{Constant False Alarm Rate (CFAR)} property with respect to noise or interference power variations.
\end{lemma}

\vspace{0.2cm}
\begin{remark}[Rank-One Target Response]
For a point-target model, $\mathbf{G}\mathbf{W}\mathbf{W}^H\mathbf{G}^H$ is rank-one, generating a single dominant eigenvalue in $\widehat{\boldsymbol{\Sigma}}$ while the remaining $N_r-1$ correspond to noise.  
Hence, for analytical tractability and without loss of generality, we focus on $N_r=2$ to capture the essential SCN behavior.
\end{remark}
\vspace{0.2cm}
\section{Detection Performance of the SCN $\kappa(\widehat{\boldsymbol{\Sigma}})$}
\label{sec:analytical_expression}
We now characterize the performance of the SCN detector by deriving closed-form expressions for its false-alarm and detection probabilities associated with the test statistic.
\vspace{0.2cm}
\begin{remark}[Effective SNR and Detection Adaptation]
Under $\mathcal{H}_1$, the disturbed covariance is 
\[\boldsymbol{\Sigma}_{\mathrm{d}}=\mu\sigma_s^2\mathbf{I}_{N_r}+\mathbf{G}\mathbf{W}\mathbf{W}^H\mathbf{G}^H,\]
yielding an \emph{effective SNR}
\begin{equation}\label{eq:eff_snr}
\gamma_e \triangleq \frac{\|\mathbf{G}\mathbf{W}\|_F^2}{\left(\mu\sigma_s^2\right)}.
\end{equation}
Thus, while CFAR holds under $\mathcal{H}_0$, the detection probability under $\mathcal{H}_1$ depends on $\gamma_e$, decreasing with stronger interference ($\mu>1$).  
This relationship reflects the natural information-theoretic trade-off between robustness and sensitivity.
\end{remark}

\vspace{0.2cm}
\begin{lemma}[Closed-Form False-Alarm Probability]\label{lemma:2}
Consider the SCN test statistic $\kappa(\widehat{\boldsymbol{\Sigma}})$ in~\eqref{eq:scn_def}, 
where $\widehat{\boldsymbol{\Sigma}} \!\sim\! \mathcal{CW}_2(L,\mathbf{I}_2)$ under $\mathcal{H}_0$.  
For $\tau>1$, the false-alarm probability \(P_{\mathrm{F}}(\tau)
\triangleq \Pr\!\left(\kappa(\widehat{\boldsymbol{\Sigma}})>\tau \mid \mathcal{H}_0\right)\) admits the closed-form expression
\begin{equation}\label{eq:cdf_integral}
\begin{split}
P_{\mathrm{F}}(\tau)
&= 1 - \!\left(\frac{2\,\Gamma\!\left(L+\tfrac{1}{2}\right)}
{\sqrt{\pi}\,\Gamma(L+1)} - 1\right)
\\[-2pt]
&\quad\times
\left(\frac{L(1-\tau)
+ 2\!\left[{}_2F_1(1,-L,L,-\tau)-1\right]}
{(4\tau)^{1-L}(1+\tau)^{2L-1}}\right),
\end{split}
\end{equation}
where ${}_2F_1(\cdot)$ denotes the Gaussian hypergeometric function and $\Gamma(\cdot)$ the Euler Gamma function.
\end{lemma}
\vspace{0.2cm}
\begin{IEEEproof}
Under $\mathcal{H}_0$, $\widehat{\boldsymbol{\Sigma}}$ follows a central complex Wishart distribution 
$\mathcal{CW}_2(L,\mathbf{I}_2)$ whose eigenvalue ratio $\kappa$ is invariant to covariance scaling.  
The c.d.f.\ $F_\kappa(\tau;\mathcal{H}_0)$ can be expressed using the joint eigenvalue density of $\mathcal{CW}_2(L,\mathbf{I}_2)$~\cite{nafkaSCN2020}.  
Applying variable transformation and evaluating the integral via Euler’s Beta function and the Gauss hypergeometric identity yields~\eqref{eq:cdf_integral}, which holds for all $L\!\ge\!2$.
\end{IEEEproof}

\vspace{3cm}
\begin{theorem}[Closed-Form Detection Probability]\label{th:1}
Consider the SCN statistic $\kappa(\widehat{\boldsymbol{\Sigma}})$ in~\eqref{eq:scn_def}, 
where $\widehat{\boldsymbol{\Sigma}}\!\sim\!\mathcal{CW}_2(L,\mathbf{I}_2,\boldsymbol{\Omega})$ with non-centrality matrix 
$\boldsymbol{\Omega} = \mathbf{G}\mathbf{W}\mathbf{W}^H\mathbf{G}^H/(\mu\sigma_s^2)$ under $\mathcal{H}_1$.  
For $\tau>1$, the detection probability 
\[
P_{\mathrm{D}}(\gamma_e,\tau)
\triangleq \Pr\!\left(\kappa(\widehat{\boldsymbol{\Sigma}})> \tau \mid \mathcal{H}_1\right)
\]
is expressed in closed form as
\begin{equation}\label{eq:cdf_H1}
\begin{aligned}
&P_{\mathrm{D}}(\gamma_e,\tau)
= 1 - 
\frac{2\,\Gamma(2L-1)e^{-\omega_1/2}}
{\omega_1\,\Gamma(L-1)^2}
\sum_{k=0}^{L} 
\frac{2^{-k}(-L)_k}{(L-1)_k\,k!}
\\
&\times
\Big(
\Phi(L\!-\!2,\,L\!+\!k;\tau,\omega_1)
- \Phi(L\!-\!1,\,L\!+\!k\!-\!1;\tau,\omega_1)
\Big)
\end{aligned}
\end{equation}
where $\omega_1 = 2L\gamma_{e}$ and the auxiliary function is defined as
\begin{equation*}\label{eq:Phi_def}
\Phi(M,\delta;\tau,\omega_1)
= \sum_{m=0}^{M}
\frac{M!\left(1-\tau^{\,\delta+m}\right)
\mathrm{E}_{1-\delta-m}\!\left(\frac{-\omega_1}{2(1+\tau)}\right)}
{(-1)^{m}m!(M-m)!(1+\tau)^{\delta+m}}.
\end{equation*}
Here, $(a)_k$ denotes the Pochhammer symbol, and 
$\mathrm{E}_n(z)=\int_1^{\infty}t^{-n}e^{-zt}\,dt$ is the exponential integral of order $n$.
\end{theorem}

\begin{IEEEproof}
See Appendix~A.
\end{IEEEproof}



\begin{remark}[Analytical Novelty and Broader Significance]
To the best of our knowledge, the closed-form expression for $P_{\mathrm{D}}(\gamma_e,\tau)$ in~\eqref{eq:cdf_H1} represents the first analytical characterization of the SCN statistic under $\mathcal{H}_1$.  
Unlike prior eigenvalue-based or LRT analyses that rely on asymptotic or simulation-based approximations, it provides an exact finite-sample formulation derived from the non-central Wishart model.  
The expression captures the joint effects of signal rank, effective SNR, and covariance mismatch, thereby unifying ideal and disturbed noise conditions within a single framework.  

Beyond ISAC, this formulation applies to a wide range of communication-theoretic detection and sensing problems such as cognitive radio, radar target detection, and MIMO spectrum monitoring where robustness to noise uncertainty is critical.  
Hence, the derived $P_{\mathrm{D}}$ extends the theoretical foundation of ratio-based detection and establishes the SCN as a universal CFAR-compliant detector for modern sensing-communication networks.
\end{remark}

\section{ISAC Power Allocation}
\label{sec:power_allocation}

The SCN detector, established as a CFAR–robust sensing method under covariance mismatch, is now employed to analyze the impact of noise uncertainty on ISAC performance.  
The focus here is on \textit{detection–theoretic characterization} rather than beamforming or waveform optimization, linking sensing reliability to communication throughput through a compact, analytically tractable framework.

The total transmit power \(P\) is partitioned between communication and sensing as  
\(
\|\mathbf{W}_\mathrm{c}\|^2 = \eta P\), and \(\|\mathbf{w}_\mathrm{s}\|^2 = (1-\eta)P,
\)
where \(0\!\le\!\eta\!\le\!1\) denotes the power–splitting ratio.  
The communication subsystem achieves ergodic rate \(\bar R(\eta)\) from~\eqref{eq:ergodic_rate}, while the sensing subsystem attains detection probabilities \(P_{\mathrm{D}}(\gamma_e,\tau)\) and \(P_{\mathrm{F}}(\tau)\), where the effective SNR scales as \(\gamma_e \!\propto\! (1-\eta)P/(\mu\sigma_s^2)\).

\subsection{Optimization Formulation}

The sensing reliability is quantified by the total error probability
\(
P_\mathrm{E}(\eta,\tau)
= \tfrac{1}{2}\!\left[P_\mathrm{F}(\tau) + 1 - P_\mathrm{D}(\gamma_e(\eta),\tau)\right],
\)
which balances false alarms and missed detections using analytical results from Lemma~\ref{lemma:1} and Theorem~\ref{th:1}.  
Coupling sensing and communication through a minimum-rate constraint yields the ISAC power–allocation problem:
\begin{subequations}\label{prob:ISAC_main}
\begin{align}
\min_{\eta,\,\tau} \quad 
& P_\mathrm{E}(\eta,\tau), \label{prob:ISAC_main_obj}\\
\text{s.t.} \quad
& \bar R(\eta) \ge R_{\min}, \label{prob:ISAC_main_rate}\\
& 0 \le \eta \le 1,\;\; \tau > 1. \label{prob:ISAC_main_bounds}
\end{align}
\end{subequations}
Problem~\eqref{prob:ISAC_main} is low-dimensional and efficiently solved via nested 1-D searches, leveraging the monotonicity of \(P_{\mathrm{F}}(\tau)\) and \(P_{\mathrm{D}}(\gamma_e,\tau)\).  
The mismatch factor \(\mu\) explicitly modulates the sensing SNR, thus governing the balance between sensing reliability and communication rate.

\subsection{Solution Structure}

The optimization naturally decouples: constraint~\eqref{prob:ISAC_main_rate} depends only on communication power, while \(P_\mathrm{E}(\eta,\tau)\) depends on sensing power and threshold.  
Hence, a sequential approach is adopted: 
\begin{itemize}
    \item \textit{Step~1 (Minimum Communication Power):}  
From~\eqref{eq:ergodic_rate}, \(\bar R(\eta)\) increases monotonically with \(P_\mathrm{c}=\eta P\).  
The minimum feasible communication power \(P_\mathrm{c}^{\min}\) satisfies \(\bar R(P_\mathrm{c}^{\min})=R_{\min}\).  
If \(P_\mathrm{c}^{\min}>P\), the ISAC task is infeasible.
\item \textit{Step~2 (Sensing Power Allocation):}  
For feasible cases, the remaining power \(P_\mathrm{s}=P-P_\mathrm{c}^{\min}\) yields effective SNR  
\(\gamma_e=P_\mathrm{s}\|\mathbf{G}\|_F^2/(\mu\sigma_s^2)\).  
The sensing subsystem then minimizes \(P_\mathrm{E}(\eta,\tau)\) for this \((\gamma_e,\mu)\).
\item \textit{Step~3 (Threshold Optimization):}  
For fixed \(\gamma_e\), \(P_\mathrm{E}(\eta,\tau)\) depends only on \(\tau\).  
Since \(P_\mathrm{F}(\tau)\) decreases and \(1-P_\mathrm{D}(\gamma_e,\tau)\) increases with $\tau$, the optimal threshold  
\(\tau^\star=\arg\min_{\tau>1}P_\mathrm{E}(\eta,\tau)\)  
is found via a simple one-dimensional line search.
\end{itemize}
This formulation preserves analytical tractability while capturing the essential trade-off between rate, power, and sensing reliability under realistic ISAC noise–uncertainty conditions.

\section{Numerical Results} \label{s:num}
\begin{figure*}[t]
    \centering
    \subfloat[SCN ROC under noise uncertainty]{
        \includegraphics[width=0.315\textwidth]{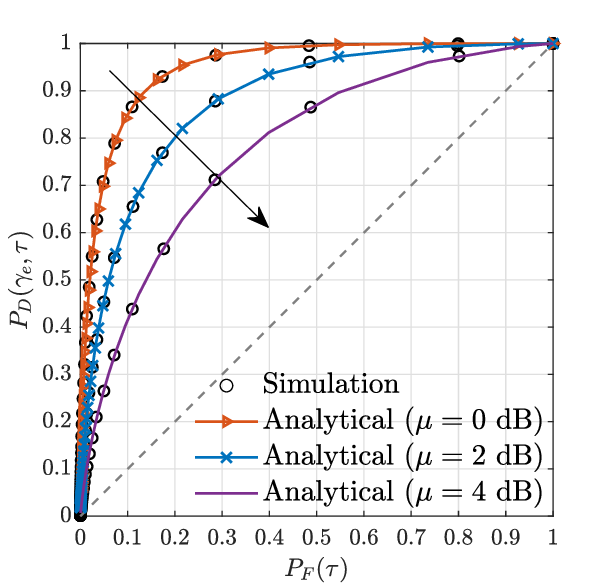}
        \label{fig:rocscn}
    }
    \hfill
    \subfloat[Error and CFAR behavior vs.\ $\mu$]{
        \includegraphics[width=0.315\textwidth]{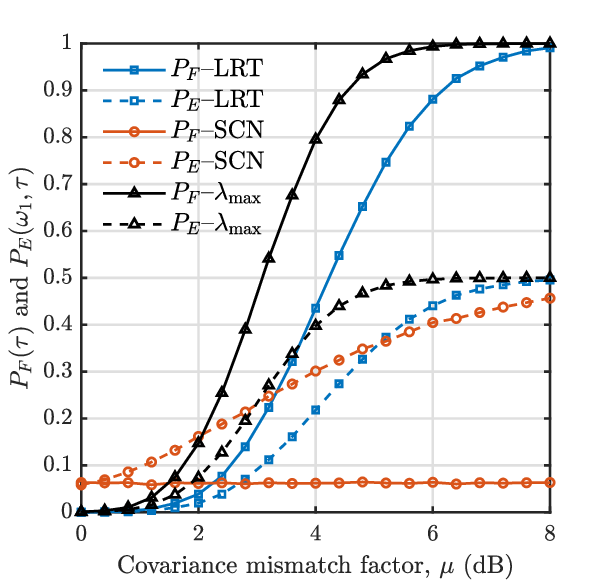}
        \label{fig:pfavsmu}
    }
    \hfill
    \subfloat[Total error vs.\ threshold $\tau$]{
        \includegraphics[width=0.315\textwidth]{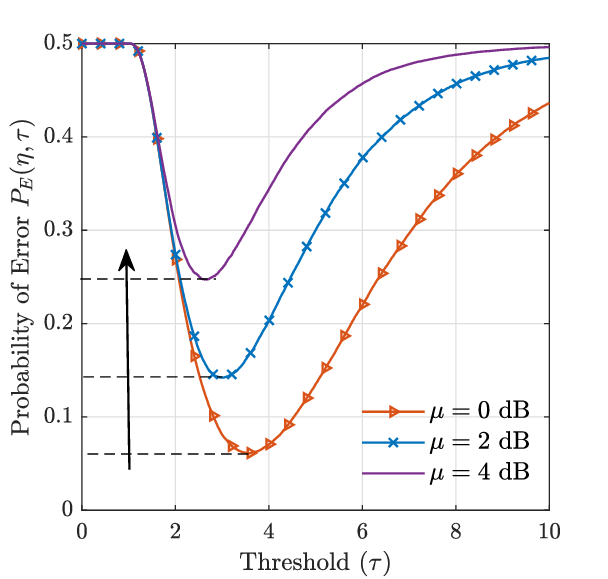}
        \label{fig:PFvTAU}
    }\vspace{-3mm}
    \caption{Performance of SCN detector, illustrating analytical–simulation agreement and CFAR robustness over benchmarks. \vspace{-6mm}}
    \label{fig:detecperf}
\end{figure*}

\begin{figure*}[t]
    \centering
    \subfloat[Achievable rate versus transmit power]{%
        \includegraphics[width=0.33\textwidth]{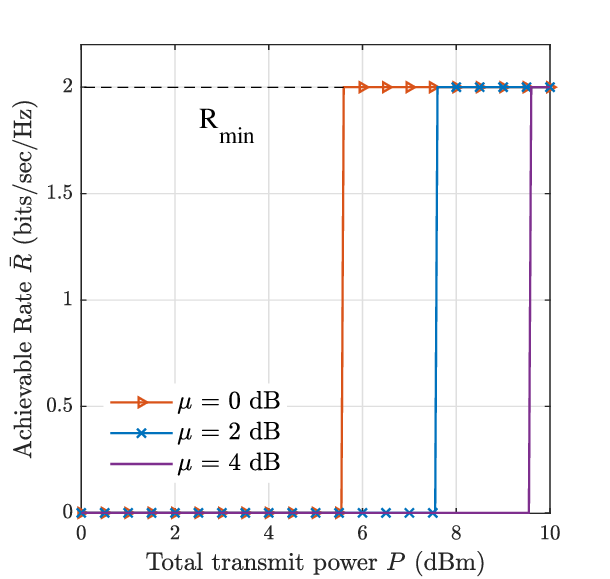}%
        \label{fig:Rvspower}} \hfill
    \subfloat[False–alarm probability versus power]{%
        \includegraphics[width=0.33\textwidth]{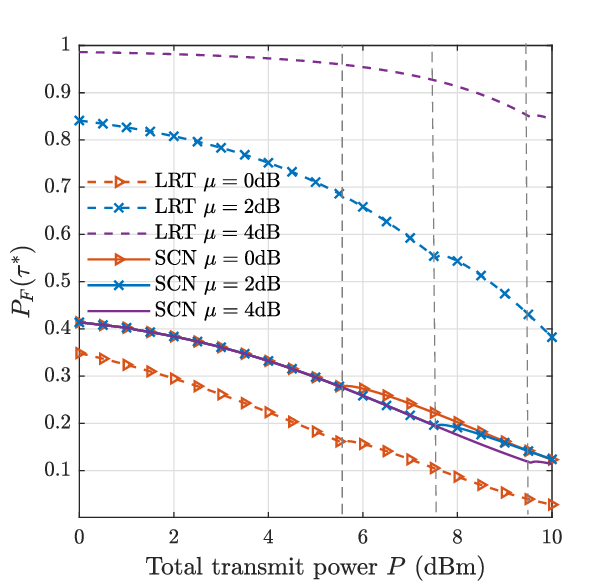}%
        \label{fig:PFvsP}} \hfill
    \subfloat[Total error probability versus power]{%
        \includegraphics[width=0.33\textwidth]{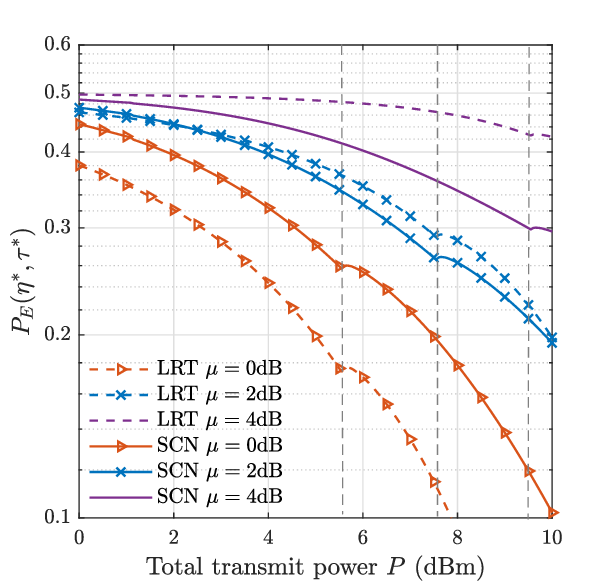}%
        \label{fig:PEvsP}}
    \caption{ISAC performance under noise uncertainty: (a) rate–power trade-off, 
    (b) CFAR robustness of SCN detector, and (c) total error comparison with LRT.\vspace{-3mm}}
    \label{fig:SenvsP}
    \vspace{-3mm}
\end{figure*}

A monostatic ISAC BS with $N_t\!=\!4$ transmit and $N_r\!=\!2$ receive antennas serves a single $N_u\!=\!4$-antenna user. 
Transmit power ranges from $0$–$10$~dBm with noise floors $\sigma_c^2\!=\!\sigma_s^2\!=\!-105$~dBm. 
The sensing channel gain $\beta$ accounts for RCS and path loss at $\theta\!=\!\pi/4$. 
All results are averaged over $10^5$ fading realizations; power and rate are expressed in dBm and bits/s/Hz.


Fig.~\ref{fig:rocscn} presents ROC curves for noise uncertainty levels $\mu\!\in\!\{0,2,4\}$~dB with $P\!=\!10$ and $\eta\!=\!0.5$. 
Analytical results from \eqref{eq:cdf_integral} and \eqref{eq:cdf_H1} closely match Monte Carlo simulations, confirming the model accuracy. 
Detection performance degrades as $\mu$ increases, where $\mu\!=\!0$~dB provides the nominal bound and $\mu\!=\!4$~dB exhibits reduced $P_\mathrm{D}$ and inflated $P_\mathrm{F}$, illustrating the sensitivity of ISAC sensing to covariance mismatch.
Fig.~\ref{fig:pfavsmu} compares SCN, LRT, and $\lambda_{\max}$ detectors under covariance mismatch $\mu$. 
For $\mu\!<\!2$~dB, $\lambda_{\max}$ yields the lowest $P_\mathrm{E}$, while for $\mu\!>\!3$~dB the SCN outperforms all, maintaining $P_\mathrm{F}\!\approx\!0.05$ regardless of $\mu$. 
Thus, SCN ensures CFAR robustness as LRT and $\lambda_{\max}$ degrade toward $P_\mathrm{E}\!\to\!0.5$ under severe mismatch.
Fig.~\ref{fig:PFvTAU} shows $P_\mathrm{E}(\eta,\tau)$ versus threshold $\tau$ for different $\mu$. 
As $\mu$ increases, curves shift upward, raising the minimum error from $P_{\mathrm{E},\min}\!\approx\!0.05$, $0.25$, and $0.45$ for $\mu\!=\!0$, $2$, and $4$~dB, with corresponding optima $\tau^*\!\approx\!5.2$, $4.8$, and $4.1$. 
These minima define the optimal $\tau^*$ minimizing $P_\mathrm{E}(\eta,\tau)$ under mismatch, validating the optimization in~\eqref{prob:ISAC_main}.

Fig.~\ref{fig:Rvspower} illustrates power allocation under the rate constraint $R_{\min}$. 
For $P_\mathrm{c}\!<\!P^*$, communication is infeasible and all power is assigned to sensing ($\eta\!=\!0$). 
Once $P_\mathrm{c}\!\ge\!P^*$, the system enters a trade-off region where excess power supports sensing while maintaining $R_{\min}$. 
As noise uncertainty $\mu$ increases, the required $P^*$ rises from $5.6$ to $7.6$ and $9.6$~dBm for $\mu\!=\!0$, $2$, and $4$~dB, respectively. 
Beyond $P^*$, $\bar{R}$ saturates at $R_{\min}$, and remaining power is adaptively redistributed via $\eta(\mu)$. 
This demonstrates efficient rate preservation and uncertainty-aware sensing–communication balance.
Fig.~\ref{fig:PFvsP} verifies the CFAR behavior of the SCN detector: 
$P_{\mathrm{F}}$ curves for $\mu\!\in\!\{0,2,4\}$~dB nearly overlap, confirming robustness to noise mismatch under $\mathcal{H}_0$. 
The minor discontinuities at threshold powers $P^*\!\in\!\{5.6,7.6,9.6\}$~dBm arise from the transition of $\eta$ from $0$ to the sensing–communication trade-off regime, reflecting power reallocation rather than noise uncertainty. 
In contrast, the LRT exhibits pronounced $P_{\mathrm{F}}$ degradation as $\mu$ increases, underscoring the SCN’s CFAR advantage.
Fig.~\ref{fig:PEvsP} illustrates the total error $P_\mathrm{E}$ versus transmit power. 
For $\mu\!=\!0$~dB, the LRT attains slightly lower $P_\mathrm{E}$ owing to perfect noise knowledge. 
However, as $\mu$ increases, SCN consistently outperforms LRT, maintaining minimal $P_\mathrm{E}$ across the entire power range. 
The widening performance gap with higher $\mu$ reveals LRT’s sensitivity to covariance mismatch and establishes SCN as a robust detector for practical ISAC environments with noise uncertainty.

\section{Conclusion}
\label{sec:conclusion}
This paper developed a unified analytical and optimization framework for \emph{Standard Condition Number} (SCN)–based detection in MIMO ISAC systems under noise uncertainty.  
Closed-form expressions for the false-alarm and detection probabilities were derived using random matrix theory, capturing the complete finite-sample behavior of SCN and revealing its invariance to covariance scaling.  
The analysis proved that the SCN inherently satisfies the \emph{constant false alarm rate} (CFAR) property, while the effective SNR degrades only as $\gamma_{e}=\gamma/\mu$ under covariance mismatch.  
Building on this characterization, a tractable ISAC power-allocation framework was formulated to minimize total detection error subject to communication rate and power constraints.  
Numerical results showed excellent agreement with theory and confirmed that SCN maintains $P_{\mathrm{F}}\!\approx\!0.05$ across $\mu\!\in\![0,4]$~dB while achieving up to $35\%$ lower total error than the LRT and $\lambda_{\max}$ detectors in strong interference conditions.  
From a communication-theoretic perspective, these results establish SCN detection as a robust and low-complexity alternative for joint sensing–communication design, offering reliability guarantees even in non-stationary and jamming-prone networks.  
Future work will extend the framework to distributed multi-target scenarios and adaptive ISAC resource allocation under real-time uncertainty.

\section*{Appendix A \\ Proof of Theorem~\ref{th:1}}
Under $\mathcal{H}_1$, the detection probability of the SCN statistic is expressed as 
$P_{\mathrm{D}}(\gamma_e,\tau)=1-\int_{1}^{\tau} f_{\kappa}(x\!\mid\!\mathcal{H}_1)\mathrm{d}x$, 
where the p.d.f.\ of the eigenvalue ratio $\kappa$ for $\widehat{\boldsymbol{\Sigma}}\!\sim\!\mathcal{CW}_2(L,\mathbf{I}_2,\boldsymbol{\Omega})$ can be written with the aid of~\cite{nafkaSCN2020} as  
\begin{equation}
\label{eq:fkappa_appendix}
\begin{split}
f_{\kappa}(x\!\mid\!\mathcal{H}_1)
&=\Psi_L(\omega_1)\frac{(x-1)x^{L-2}}{(x+1)^{2L-1}}
\\[-2pt]
&\quad\times\!\Big[\mathcal{F}_L\!\Big(\tfrac{\omega_1 x}{2(x+1)}\Big)-\mathcal{F}_L\!\Big(\tfrac{\omega_1}{2(x+1)}\Big)\Big],
\end{split}
\end{equation}
where $\omega_1=2L\gamma_e$ and 
\[\Psi_L(\omega_1)=\tfrac{2\,\Gamma(2L-1)e^{-\omega_1/2}}{\omega_1\Gamma(L-1)^2};\, \mathcal{F}_L(z)={}_1F_1(2L-1,L-1,z).\]  

Applying identity ${}_1F_1(a,b,z)\!=\!e^{z}{}_1F_1(b-a,b,-z)$ and the finite series expansion ${}_1F_1(-L,b,z)=\sum_{k=0}^{L}\frac{(-L)_k z^k}{(b)_k k!}$~\cite{gradshteyn2007book} yields 
\[\int_{1}^{\tau}\!f_{\kappa}(x\!\mid\!\mathcal{H}_1)\mathrm{d}x
=\Psi_L(\omega_1)\!\sum_{k=0}^{L}\frac{(-L)_k}{(L-1)_k k!}\mathcal{I}_k(\omega_1,L,\tau),\]
where
\begin{align*}
\mathcal{I}_k(\omega_1,L,\tau)
&=\!\int_{1}^{\tau}\!\!\frac{x^{L-1}-x^{L-2}}{(x+1)^{2L-1}}\!\Big[
e^{\frac{\omega_1 x}{2(x+1)}}\!\Big(\!\tfrac{-\omega_1 x}{2(x+1)}\!\Big)^k\\[-2pt]
&\qquad\qquad \qquad - e^{\frac{\omega_1}{2(x+1)}}\!\Big(\!\tfrac{-\omega_1}{2(x+1)}\!\Big)^k\Big]\mathrm{d}x .
\end{align*}
By substituting $y=x/(x+1)$ and $y=1/(x+1)$ for the two exponential terms, 
$\mathcal{I}_k$ is expressed as a combination of integrals of the form 
$\int e^{\frac{\omega_1 y}{2}}(-\omega_1 y/2)^k(y(1-y))^{L-2}(2y-1)\mathrm{d}y$.
Expanding $(1-y)^{L-2}$ via the binomial theorem and integrating termwise leads to 
closed forms involving $\mathrm{E}_n(z)=\int_{1}^{\infty}t^{-n}e^{-zt}\mathrm{d}t$~~\cite{gradshteyn2007book}, 
yielding the auxiliary function $\Phi(M,\delta;\tau,\omega_1)$ in Theorem~\ref{th:1}.  
Substituting back and simplifying gives the final result in~\eqref{eq:cdf_H1}, 
thus completing the proof. \hfill$\blacksquare$

\end{document}